\documentclass[conference,letterpaper]{IEEEtran}

\usepackage{graphicx}
\usepackage{color}
\usepackage{pgf, tikz, pgfplots}
\usepackage{siunitx}
\pgfplotsset{compat=1.17} 
\usepackage{amssymb} 
\usepackage{bm}
\usepackage{acronym}
\usepackage{tikz}
\usepackage{mathrsfs}
\usepackage{enumerate}
\usepackage{todonotes}
\usepackage[utf8]{inputenc} 
\usepackage[T1]{fontenc}
\usepackage{url}
\usepackage{ifthen}
\usepackage{cite}
\usepackage[cmex10]{amsmath} %

\usetikzlibrary{arrows,calc,fit,matrix,positioning,spy}
\tikzset{myblock/.style={rectangle, draw, thin, minimum width=0.6cm, minimum height=0.6cm},font=\footnotesize,align=center}%
\tikzset{mywideblock/.style={rectangle, draw, thin, minimum width=1cm, minimum height=0.6cm},font=\footnotesize,align=center}%
\newcommand{\myline}[2]{
\path(#1.east) --(#2.west)  coordinate[pos=0.4](mid);
\draw[-latex] (#1.east) -| (mid) |- (#2.west);
}

\usepackage{cite}
\usepackage{amsmath,amssymb,amsfonts}

\usepackage{dsfont}
\usepackage{algorithm}
\usepackage[noend]{algorithmic}
\usepackage{graphicx}
\usepackage{textcomp}
\usepackage{xcolor}
\def\BibTeX{{\rm B\kern-.05em{\sc i\kern-.025em b}\kern-.08em
    T\kern-.1667em\lower.7ex\hbox{E}\kern-.125emX}}

    \usepackage{amsfonts}     
\usepackage{graphicx}
\usepackage{color}
\usepackage{pgf, tikz, pgfplots}
\usepackage{siunitx}
\pgfplotsset{compat=1.17} 
\usepackage{amssymb} 
\usepackage{bm}
\usepackage{bbm}
\usepackage{amsthm}
\usepackage{acronym}
\usepackage{tikz}
\usepackage{mathrsfs}
\usepackage{enumerate}
\usetikzlibrary{arrows,calc,fit,matrix,positioning,shapes,shadows,trees,mindmap,tikzmark,arrows.meta,angles,quotes,babel,spy}

\usepackage{booktabs}

\definecolor{KITblack}{rgb}{0.25,0.25,0.25}
\definecolor{KITbrown}{rgb}{0.65,0.51,0.18}

\definecolor{kit-green100}{rgb}{0,.59,.51}
\definecolor{kit-green100}{rgb}{0,.59,.51}
\definecolor{kit-green70}{rgb}{.3,.71,.65}
\definecolor{kit-green50}{rgb}{.50,.79,.75}
\definecolor{kit-green30}{rgb}{.69,.87,.85}
\definecolor{kit-green15}{rgb}{.85,.93,.93}
\definecolor{KITgreen}{rgb}{0,.59,.51}

\definecolor{KITpalegreen}{RGB}{130,190,60}
\colorlet{kit-maigreen100}{KITpalegreen}
\colorlet{kit-maigreen70}{KITpalegreen!70}
\colorlet{kit-maigreen50}{KITpalegreen!50}
\colorlet{kit-maigreen30}{KITpalegreen!30}
\colorlet{kit-maigreen15}{KITpalegreen!15}

\definecolor{KITblue}{rgb}{.27,.39,.66}
\definecolor{kit-blue100}{rgb}{.27,.39,.67}
\definecolor{kit-blue70}{rgb}{.49,.57,.76}
\definecolor{kit-blue50}{rgb}{.64,.69,.83}
\definecolor{kit-blue30}{rgb}{.78,.82,.9}
\definecolor{kit-blue15}{rgb}{.89,.91,.95}

\definecolor{KITyellow}{rgb}{.98,.89,0}
\definecolor{kit-yellow100}{cmyk}{0,.05,1,0}
\definecolor{kit-yellow70}{cmyk}{0,.035,.7,0}
\definecolor{kit-yellow50}{cmyk}{0,.025,.5,0}
\definecolor{kit-yellow30}{cmyk}{0,.015,.3,0}
\definecolor{kit-yellow15}{cmyk}{0,.0075,.15,0}

\definecolor{KITorange}{rgb}{.87,.60,.10}
\definecolor{kit-orange100}{cmyk}{0,.45,1,0}
\definecolor{kit-orange70}{cmyk}{0,.315,.7,0}
\definecolor{kit-orange50}{cmyk}{0,.225,.5,0}
\definecolor{kit-orange30}{cmyk}{0,.135,.3,0}
\definecolor{kit-orange15}{cmyk}{0,.0675,.15,0}

\definecolor{KITred}{rgb}{.63,.13,.13}
\definecolor{kit-red100}{cmyk}{.25,1,1,0}
\definecolor{kit-red70}{cmyk}{.175,.7,.7,0}
\definecolor{kit-red50}{cmyk}{.125,.5,.5,0}
\definecolor{kit-red30}{cmyk}{.075,.3,.3,0}
\definecolor{kit-red15}{cmyk}{.0375,.15,.15,0}

\definecolor{KITpurple}{RGB}{160,0,120}
\colorlet{kit-purple100}{KITpurple}
\colorlet{kit-purple70}{KITpurple!70}
\colorlet{kit-purple50}{KITpurple!50}
\colorlet{kit-purple30}{KITpurple!30}
\colorlet{kit-purple15}{KITpurple!15}

\definecolor{KITcyanblue}{RGB}{80,170,230}
\colorlet{kit-cyanblue100}{KITcyanblue}
\colorlet{kit-cyanblue70}{KITcyanblue!70}
\colorlet{kit-cyanblue50}{KITcyanblue!50}
\colorlet{kit-cyanblue30}{KITcyanblue!30}
\colorlet{kit-cyanblue15}{KITcyanblue!15}

\newcommand\blfootnote[1]{%
    \begingroup
    \renewcommand\thefootnote{}\footnote{#1}%
    \addtocounter{footnote}{-1}%
    \endgroup
}

\newcommand\extrafootertext[1]{%
    \bgroup
    \renewcommand\thefootnote{\fnsymbol{footnote}}%
    \renewcommand\thempfootnote{\fnsymbol{mpfootnote}}%
    \footnotetext[0]{#1}%
    \egroup
}

\usepackage{url}
\usepackage[colorlinks=true,bookmarks=false,citecolor=black,urlcolor=black,linkcolor=black]{hyperref} 
\usepackage[hyphenbreaks]{breakurl}

\newtheorem{theorem}{Theorem}

\newtheorem{lemma}{Lemma}

\begin{document}
\title{%
Subcode Ensemble Decoding of Linear Block Codes\\
}

\author{\IEEEauthorblockN{Jonathan Mandelbaum, Holger Jäkel, and Laurent Schmalen}
\IEEEauthorblockA{Communications Engineering Lab, Karlsruhe Institute of Technology (KIT), 76131 Karlsruhe, Germany\\
\texttt{jonathan.mandelbaum@kit.edu}}
}

\maketitle

\begin{abstract}
Low-density parity-check (LDPC) codes together with belief propagation (BP) decoding yield exceptional error correction capabilities in the large block length regime. Yet, there remains a gap between BP decoding and maximum likelihood decoding for short block length LDPC codes. In this context, ensemble decoding schemes yield both reduced latency and good error rates.
In this paper, we propose subcode ensemble decoding (SCED), which employs an ensemble of decodings on different subcodes of the code. To ensure that all codewords are decodable, we use the concept of linear coverings and explore approaches for sampling suitable ensembles for short block length LDPC codes. 
Monte-Carlo simulations conducted for three LDPC codes demonstrate that SCED improves decoding performance compared to stand-alone decoding and automorphism ensemble decoding. In particular,
in contrast to existing schemes, e.g., multiple bases belief propagation and automorphism ensemble decoding, SCED does not require the NP-complete search for low-weight dual codewords or knowledge of the automorphism group of the code, which is often unknown.
\end{abstract}
\section{Introduction}
\vspace*{-1.25em}
\blfootnote{This work has received funding from the 
German Federal Ministry of Education and Research (BMBF) within the project Open6GHub (grant agreement 16KISK010) and the European Research Council (ERC) under the European Union’s Horizon 2020 research and innovation programme (grant agreement No. 101001899).}

In the large block length regime, low-density parity-check (LDPC) codes provide exceptional error-correction performance when decoded with a low-complexity message passing algorithm, often denoted belief propagation (BP) decoding \cite{MCT08}.
However, emerging applications for 6G, such as ultra-reliable low-latency and machine-type communications, demand short block length codes \cite{8594709,10571997}.
In this regime, LDPC codes with BP decoding exhibit a performance gap compared to maximum likelihood~(ML) decoding. Bridging this gap is essential in the search for a unified coding scheme \cite{10274812}, i.e., a single code family that performs well across all block lengths. %

Ensemble decoding has demonstrated performance improvements in this regime by improving both latency and error-correction capabilities\cite{AED_RMcodes,MBBP1}. 
They are based on the observation that, for any binary-input memoryless symmetric output channel, the error probability of message-passing decoding depends on the underlying graphical representation of the code 
and the noise introduced by the channel rather than on the transmitted codeword itself \cite[Lemma 4.90]{MCT08}.
Ensemble decoding schemes exploit this property by using either an ensemble of varied noise representations or of altered graphs for decoding\cite{AED_RMcodes,geiselhart2023ratecompatible,stuttgart_ldpc_aed,8723089,7360541,MBBP1,MBBP2,MBBP3_withLeaking,mandelbaum2024endomorphisms,mandelbaum2023generalized}.
On the one hand, automorphism ensemble decoding (AED), generalized AED (GAED), and endomorphism ensemble decoding (EED) use knowledge of the code structure to alter the noise representation.
For instance, assuming the transmission of a codeword ${\bm{x}\in \mathcal{C}}$ over an additive white Gaussian noise (AWGN) channel, AED uses an automorphism $\pi$ to map the received word $\bm{y}=\bm{x}+\bm{n}$ onto $\pi(\bm{y})=\pi(\bm{x})+\pi(\bm{n})$, i.e., onto a possibly different codeword $\pi(\bm{x})\in\mathcal{C}$ superimposed with an altered noise representation $\pi(\bm{n})$.
If the automorphism group is known and the effect of its automorphisms are not absorbed by the symmetry of the decoder, AED can improve decoding performance \cite{AED_RMcodes,geiselhart2023ratecompatible,stuttgart_ldpc_aed}. 
 Other schemes that alter the noise representation are noise-aided ensemble decoding (NED) \cite{8723089} and saturated belief propagation (SBP) \cite{7360541}.
On the other hand, multiple bases belief propagation (MBBP)
and scheduling ensemble decoding (SED)
improve performance by leveraging ensembles of different graphs or variants of BP. For instance, MBBP uses a set of equivalent parity-check matrices (PCMs), i.e., different graphs,
while SED varies the scheduling of its constituent layered decodings \cite{MBBP1,krieg2024comparativestudyensembledecoding_arxiv}.
Yet, MBBP requires the NP-complete search for low-weight dual codewords, severely limiting its application\cite{intractabilityofminimumdistance,kraft_setmatrix}.%

In this work, we propose subcode ensemble decoding (SCED), a scheme that uses a set of decodings on different subcodes of the code. SCED employs a linear covering in which the union of those subcodes cover the code \cite{clark_covering_numbers}, to ensure that all codewords are decodable. 
We introduce the \emph{relative coverage} to compare ensemble decoding schemes and to simplify the search for suitable ensembles of a given size.
In particular, we demonstrate that generating effective ensembles for SCED of short LDPC codes is straightforward without requiring knowledge of the code structure beyond its PCM or any specific BP decoding requirements (e.g., layered scheduling).
Monte-Carlo simulations show that SCED yields improved error correction capabilities compared to AED.
Notably, SCED results in gains in terms of error probability compared to stand-alone decoding for a fixed total number of iterations,
while offering reduced latency, due to its parallelizable structure in combination with a reduced maximum number of iterations per constituent BP decoding.

\section{Preliminaries}

In this work, we consider binary linear block codes $\mathcal{C}(n,k)$ which are $k$-dimensional subspaces of $\mathbb{F}_2^n$.
The parameters $n\in \mathbb{N}$ and $k\in \mathbb{N}$ denote the block length and information length, respectively, and are omitted when clear from the context. Linear block codes can be defined via their non-unique parity\--check matrix (PCM) $\bm{H}\in \mathbb{F}_2^{m\times n}$, such that
$$\mathcal{C}\left(n,k\right)=\left\{\bm{x}\in \mathbb{F}_2^n:\bm{H} \bm{x} = \bm{0}\right\}=\mathrm{Null}(\bm{H}).$$
A linear subcode ${\mathcal{C}_\mathrm{s}(n,k')}$ of a code $\mathcal{C}(n,k)$, denoted as ${\mathcal{C}_\mathrm{s}\subseteq\mathcal{C}}$, is a $k'$-dimensional subspace of $\mathcal{C}$, where we assume $k'\leq k$.
For a proper linear subcode, $\subseteq$ is replaced by $\subset$.

LDPC codes are linear block codes characterized by sparse PCMs. They are typically decoded using variants of BP, such as the sum-product algorithm (SPA) and the (normalized) min-sum algorithm (MSA)\cite{MCT08}.
BP decoding operates on the Tanner graph of the code, where messages---typically represented as log-likelihood ratios (LLRs)---are iteratively exchanged along its edges. The Tanner graph is a bipartite graph representation of a PCM comprising two disjoint sets of vertices: variable nodes (VNs) and check nodes (CNs). The VN $\mathrm{v}_j$ corresponds to column $j$ of the PCM, representing a code bit, while CN $\mathrm{c}_i$ corresponds to row $i$ of the PCM, representing a parity check. VN $\mathrm{v}_j$ is connected to CN $\mathrm{c}_i$ if $H_{i,j}=1$ \cite{MCT08}, where $H_{i,j}$ denotes the element in row $i$ and column $j$ of $\bm{H}$.

\section{Subcode Ensemble Decoding} \label{sec:sced}
\begin{figure}
    \centering
    \begin{tikzpicture}	  
\DeclareRobustCommand{\svdots}{%
  \vbox{%
    \baselineskip=0.33333\normalbaselineskip
    \lineskiplimit=0pt
    \hbox{.}\hbox{.}\hbox{.}%
    \kern-0.2\baselineskip
  }%
}

\tikzset{mywideblock/.style={rectangle, draw, thin, minimum width=0.97cm, minimum height=0.5cm},font=\footnotesize,align=center}%
    \node (input) at (-1,-0.6) {$\bm{y}$};    
    \draw[fill] (-0.19,-0.6) circle (1pt);

    \node (debox_1) [mywideblock] at (1.2,0 )  {$\mathrm{Dec}_1$};
    \node (debox_K) [mywideblock] at (1.2,-1.2)  {$\mathrm{Dec}_K$};

    \node (ml_1) [right=1cm of debox_1] {};
    \node (ml_K) [right=1cm of debox_K] {};

    \node (ml_box) [rectangle, draw, thin, minimum width=0.6cm, minimum height=1.6cm] at (3,-0.6)  {\rotatebox{90}{ML-in-the-list}};
    \node (output) [right=0.3cm of ml_box] {$\hat{\bm{x}}$};

    \node (dots_box) [below=0.1cm of debox_1]  {$\svdots$};
    \myline{input}{debox_1};
    \myline{input}{debox_K};

    \draw [-latex] (debox_1) -- (ml_1);
    \draw [-latex] (debox_K) -- (ml_K);

    \draw [-latex] (ml_box) -- (output);	

\end{tikzpicture}
    \caption{Block diagram of SCED using $K$ different subcodes $\mathcal{C}_i\subseteq \mathcal{C}$ and their respective decoder $\mathrm{Dec}_i$.  }
    \label{figure:sced}
\vspace*{-1.5em}
\end{figure}
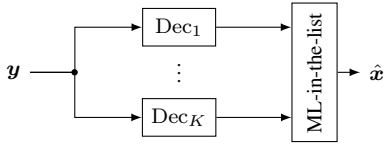

Fig.~\ref{figure:sced} 
 depicts the block diagram of subcode ensemble decoding (SCED) of an arbitrary linear code $\mathcal{C}$. We consider the transmission of a codeword $\bm{x}\in \mathcal{C}$ over an arbitrary channel with output alphabet $\mathcal{Y}$. The receiver observes ${\bm{y}\in \mathcal{Y}^n}$ which is then used as input to $K$ parallel decoding algorithms ${\mathrm{Dec}_i:\mathcal{Y}^n\rightarrow \mathbb{F}^n,i\in [K]:=\{1,2,\ldots,K\}}$, called \emph{paths}, yielding $K$ estimates $\hat{\bm{x}}_i=\mathrm{Dec}_i(\bm{y})$. In contrast to other ensemble decoding schemes, SCED possibly uses proper subcodes ${\mathcal{C}_i\subset \mathcal{C}}$ and their respective decoding in every path.
 The final estimate $\hat{\bm{x}}$ of SCED is chosen according to an {\emph{ML-in-the-list-rule}}\cite{AED_RMcodes}:
\begin{align*}
    \hat{\bm{x}}:= \arg \max_{\bm{x}\in \mathcal{L}} L(\bm{y}|\bm{x}),
\end{align*}
where $L(\bm{y}|\bm{x})$ denotes the log-likelihood and with list
\[
\mathcal{L}:=\begin{cases}
    \{\hat{\bm{x}}_i:i\in[K],\hat{\bm{x}}_i\in \mathcal{C}\},& \text{if } \exists i\in[K]:\hat{\bm{x}}_i\in \mathcal{C} \\
    \{\hat{\bm{x}}_i:i\in[K]\},& \text{otherwise.}
\end{cases}
\]

To potentially decode all codewords, the ensemble of subcodes must constitute a \emph{linear covering} (LC).
Following the definition for subspaces of vector spaces in \cite{clark_covering_numbers}, we define the LC of a code $\mathcal{C}$ as a set of subcodes $\{\mathcal{C}_i:i\in {[K]}\}$ with 
\begin{equation}
\bigcup_{i=1}^K \mathcal{C}_i= \mathcal{C}.\label{eq:general_cover}
\end{equation}

\section{Approaches to Choose Suitable Subcodes}\label{sec:choose_sub}
From now on, we consider SCED with all path decodings being BP decoding and we assume ${\mathcal{C}_1= \mathcal{C}}$ such that ${\mathrm{Dec}_1=\mathrm{BP}(\bm{H})}$, i.e., the additional paths complement stand-alone BP decoding.
Given a PCM $\bm{H}$, a PCM $\bm{H}_\ell$ of a subcode ${\mathcal{C}_\ell\subset\mathcal{C}}$ can be obtained via 
\vspace*{-0.5em}
\begin{equation}\label{eq:inducing_subcode}
    \bm{H}_\ell=\begin{pmatrix}
    \bm{H}\\
    \bm{h}_\ell
\end{pmatrix},
\end{equation}
i.e., by appending at least one arbitrary row $\bm{h}_\ell\in\mathbb{F}_2^{1\times n}$ to $\bm{H}$.
In the following, we say that $\bm{H}_\ell$ \emph{induces a subcode} to refer to the process of constituting the subcode $\mathcal{C}(\bm{H}_\ell)\subseteq\mathcal{C}(\bm{H})$. 

Note that if $\bm{h}_\ell$ is linearly independent of the rows of $\bm{H}$,  using (\ref{eq:inducing_subcode}) $\bm{H}_\ell$ induces a proper subcode $\mathcal{C}_\ell \subset \mathcal{C}$ with ${\mathrm{dim}(\mathcal{C}_\ell)=k-1}$.
Hence, SCED simplifies the search for rows compared to MBBP by allowing to append rows that are possibly linearly independent of the rows of $\bm{H}$ avoiding the NP-complete search for low-weight dual codewords.
Similarly, by simply appending rows, we avoid searching for sparse PCMs for each subcode and benefit from the originally designed PCMs of LDPC codes, which are well-suited for BP.

\emph{Remark:} Note that by appending a matrix $\bm{\mathcal{H}}_\ell\in \mathbb{F}_2^{m_\ell\times n}$ rather than the row $\bm{h}_\ell$ in (\ref{eq:inducing_subcode}), the approach can be generalized to use lower dimensional subcodes. Our simulations show promising results for ${m_\ell=1}$ such that we constrain ourselves to append row vectors, i.e, $\bm{\mathcal{H}}_\ell=\bm{h}_\ell$, in the following.

\subsection{Suitable Subcodes for BP Decoding}\label{sec:subcodes_for_bp}
The performance of BP decoding depends on the structure of the Tanner graph. Typically, graphs related to sparse PCMs with no (or only a few) 4-cycles yield the best decoding performance.
With this in mind, we investigate two approaches to sample the rows $\bm{h}_\ell,\,\ell \in [K]\setminus\{1\}$, i.e., the rows that, using (\ref{eq:inducing_subcode}) induce the paths that complement stand-alone BP.
Either, we sample the elements of $\bm{h}_\ell$ according to a Bernoulli distribution ${(\bm{h}_\ell)_{i}\overset{\mathrm{iid}}{\sim} \mathcal{B}(p)}$. 
We choose a small probability $p$ to obtain sparse rows. If $\bm{h}_\ell=\bm{0}$, we resample $\bm{h}_\ell$ as an all-zero $\bm{h}_\ell$ has no effect on the decoding behavior.
Or, we sample rows $\bm{h}_\ell$ with $\mathrm{w}_\mathrm{H}(\bm{h}_\ell)=d_\mathrm{c}$, where $\mathrm{w}_\mathrm{H}$ denotes the Hamming weight, and which introduce no new 4-cycles%
, i.e., the number of {4-cycles} of $\bm{H}_\ell$ equals the number of 4-cycles in $\bm{H}$.

\subsection{Linear Covering of Codes}
By employing decodings of the subcodes, not all paths can recover the transmitted codeword $\bm{x}$ since possibly $\bm{x}\notin\mathcal{C}_i$.
Therefore, the question arises if it is beneficial to choose ${\{\mathcal{C}_i:i\in[K]\setminus\{1\}\}}$ fulfilling (\ref{eq:general_cover}),
i.e., ${\forall\bm{x}\in\mathcal{C}}, \exists{ \ell \in[K]\setminus\{1\}}: \bm{x}\in\mathcal{C}_\ell $.
 This design goal appears reasonable because the error probability of BP decoding and many of its variants is independent of the transmitted codeword. 
 Note that the smallest number of proper linear subcodes to constitute an LC is $3$ \cite{clark_covering_numbers}.

\begin{theorem}\label{theorem:lc_k-1_dim} Let ${\bm{h}_1,\bm{h}_2\in \mathbb{F}^{1\times n}}$ be two row vectors that are linearly independent of the rows of a PCM $\bm{H}$ of $\mathcal{C}$ and let ${\bm{h}_3=\bm{h}_2+\bm{h}_1}$.
Using (\ref{eq:inducing_subcode}), they induce subcodes ${\mathcal{C}_1,\mathcal{C}_2\subset \mathcal{C}},$ and $\mathcal{C}_3\subseteq\mathcal{C}$, respectively, such that $\{\mathcal{C}_1,\mathcal{C}_2,\mathcal{C}_3\}$ fulfills (\ref{eq:general_cover}).
\end{theorem}
Note that in Theorem~\ref{theorem:lc_k-1_dim} only $\mathcal{C}_1$ and $\mathcal{C}_2$ are required to be proper subcodes.
Theorem~\ref{theorem:lc_k-1_dim} is proven in the appendix.
Based on Theorem~\ref{theorem:lc_k-1_dim}, we can find an LC by randomly sampling $\bm{h}_1$ and $\bm{h}_2$ and calculating $\bm{h}_3=\bm{h}_1+\bm{h}_2$.

\begin{algorithm}[bt!]
\caption{Generation of $4$-cycle free rows}\label{alg:append_row}
\begin{algorithmic}[1]
\STATE \textbf{Input:} Row weight $d_\mathrm{c}$, parity-check matrix $\bm{H}$, $\mathcal{F}_\mathrm{old}$
\STATE \textbf{Output:} New row $\bm{h}$ of weight $d_\mathrm{c}$ or failure
\STATE Initialize $\mathcal{F} \gets [n]$ or $\mathcal{F}_\mathrm{old}$
\STATE Initialize $\bm{h} \gets \bm{0}$
\WHILE{$\mathrm{w_H}(\bm{h}) < d_\mathrm{c}$}
    \IF{$\mathcal{F} = \emptyset$}
        \STATE \textbf{Return} failure
    \ENDIF
    \STATE Randomly select $j \in \mathcal{F}$
    \STATE $h_j \gets 1$
    \STATE $\mathcal{F} \gets \mathcal{F}\setminus \bigcup_{i\in[m]:H_{i,j}=1}\left\{\ell\in[n]:H_{i,\ell}=1\right\}$%
\ENDWHILE
\STATE \textbf{Return} $\bm{h}$
\end{algorithmic}
\end{algorithm}

We propose Algorithm~\ref{alg:append_row} to generate a row $\bm{h}$ with row weight $d_\mathrm{c}$ that does not contribute new $4$-cycles to $\bm{H}$.
 Algorithm~\ref{alg:append_row} employs a set $\mathcal{F}$ which contains feasible indices of non-zero positions of $\bm{h}$ that yield no additional 4-cycle.
It iteratively selects a random ${j\in\mathcal{F}}$ and sets ${h_j=1}$ until ${\mathrm{w_H}(\bm{h})=d_\mathrm{c}}$. After choosing ${j\in\mathcal{F}}$, 
all indices of possible non-zero positions that would constitute a 4-cycle together with $h_j$
 are removed from $\mathcal{F}$.
Those indices are given by 
 $$\bigcup_{i\in[m]:H_{i,j}=1}\left\{\ell\in[n]:H_{i,\ell}=1\right\}.$$
If $\mathcal{F}=\emptyset$ before $\mathrm{w_H}(\bm{h})=d_\mathrm{c}$ is obtained, the algorithm fails to find a new row, requiring the procedure to be repeated.

To construct rows that do not necessarily constitute an LC, we initialize $\mathcal{F}=[n]$. To sample three rows that result in an LC and do not contribute new 4-cycles, we first run Algorithm~\ref{alg:append_row} for $\bm{h}_1$ and use the resulting $\mathcal{F}$, denoted as $\mathcal{F}_\mathrm{old}$, as initialization for Algorithm~\ref{alg:append_row} when sampling $\bm{h}_2$. Then, appending $\bm{h}_3=\bm{h}_1+\bm{h}_2$ to $\bm{H}$ yields no new 4-cycles. Note that this approach results in $\mathrm{w_H}(\bm{h}_3)=2\cdot d_\mathrm{c}$.
Algorithm~\ref{alg:append_row} is greedy and is not guaranteed to find a new row.
By iteratively applying Algorithm~\ref{alg:append_row}, it is possible to append multiple rows to $\bm{H}$ to constitute even smaller dimensional subcodes.

\subsection{Maximum-Coverage Heuristic}

Next, we are interested in selecting a good ensemble consisting of $K$ paths out of a larger set of candidate subcodes induced by different appended rows. 
Notably, such a subset does not necessarily consist of those candidates with the best stand-alone decoding performance, but those collaborating in the best way. 
To identify a subset of $K$ candidate paths that yields the best performance in ensemble decoding, we follow the approach used in~\cite{kraft_setmatrix}: select $c\in\mathbb{N}$ candidates paths and
 simulate the transmission of $N$ fixed but arbitrary frames at a specific SNR for all candidates. 
 Next, we construct sets $\mathcal{S}_i\subseteq [N]$, with $i\in[c]$, where $j\in\mathcal{S}_i$ if the $i$th candidate successfully decodes the $j$th received word.
 The task of finding a good ensemble can be modeled as a maximum coverage problem\cite{kraft_setmatrix}, which aims to find a subset $\mathcal{E}\subset\{S_1,\ldots, S_c\}$ of cardinality $|\mathcal{E}|=K$ maximizing $|\bigcup_{S\in\mathcal{E}}S|$\cite{10.1007/3-540-48777-8_2}.
We refine the approach in \cite{kraft_setmatrix} as follows:
we carry out BP decoding using $\bm{H}$ at an SNR yielding an FER of $10^{-3}$ until $N$ frame errors are accumulated.
Then, $j\in \mathcal{S}_i$ if $\mathrm{Dec}_i(\bm{y}_j)=\bm{x}_j$, where $\bm{x}_j$ and $\bm{y}_j$ denote the $j$th transmitted codeword and received frame, respectively,
i.e., we select $\Tilde{K}:=K-1$ paths that correct as many frames as possible when BP decoding on $\bm{H}$ fails. Together with decoding on $\bm{H}$, they constitute an ensemble of size $K$.
Since the maximum-coverage problem is NP-hard, we use the algorithm presented in \cite[Algorithm 3]{kraft_setmatrix}.%

\subsection{Comparison}\label{sec:comparison_ensemble}
\begin{figure}
    \centering
    \begin{tikzpicture}[scale=0.92,spy using outlines={rectangle, magnification=2}]

\begin{axis}[%
width=.8\columnwidth,
height=3.7cm,
at={(0.758in,0.645in)},
scale only axis,
xmin=1,
xmax=4,
xlabel style={font=\color{white!15!black}},
xlabel={$E_{\mathrm{b}}/N_0$ ($\si{dB}$)},
ymode=log,
ymin=1e-04,
ymax=1,
yminorticks=true,
ylabel style={font=\color{white!15!black}},
ylabel={FER},
axis background/.style={fill=white},
xmajorgrids,
ymajorgrids,
legend style={at={(0.03,0.03)}, anchor=south west, legend cell align=left, align=left, draw=white!15!black,font=\scriptsize}
]

\addplot[color=KITgreen,line width = 1pt,mark=x, mark options={solid}]
table[row sep=crcr]{
 1.00  5.076142e-01\\
 1.50  3.134796e-01\\
 2.00  1.434720e-01\\
 2.50  5.898835e-02\\
 3.00  1.854513e-02\\
 3.50  4.438428e-03\\
 4.00  9.914537e-04\\
};
\addlegendentry{SPA};

\addplot[color=KITpurple,line width = 1pt,mark=triangle, mark options={solid}]
table[row sep=crcr]{
 1.00  4.901961e-01\\
 1.50  2.490660e-01\\
 2.00  1.181335e-01\\
 2.50  4.115226e-02\\
 3.00  1.263025e-02\\
 3.50  2.587858e-03\\
 4.00  4.288145e-04\\
};
\addlegendentry{Ensemble 1};

\addplot[color=KITpurple!70,line width = 1pt,mark=square, mark options={solid}]
table[row sep=crcr]{
 1.00  4.608295e-01\\
 1.50  2.583979e-01\\
 2.00  1.185536e-01\\
 2.50  4.988153e-02\\
 3.00  1.183152e-02\\
 3.50  2.913880e-03\\
 4.00  5.224647e-04\\
};
\addlegendentry{Ensemble 2};

\addplot[color=KITpurple!40,line width = 1pt,mark=*, mark options={solid}]
table[row sep=crcr]{
 1.00  4.914005e-01\\
 1.50  2.797203e-01\\
 2.00  1.262626e-01\\
 2.50  4.230566e-02\\
 3.00  1.330318e-02\\
 3.50  2.631995e-03\\
 4.00  5.180555e-04\\
};
\addlegendentry{Ensemble 3};

\end{axis}

\end{tikzpicture}%
    \caption{ Performance of SCED using the different ensembles and stand-alone decoding for the 5G LDPC code $\mathcal{C}_{5\mathrm{G}}(132,66)$.}
    \label{fig:compare_construction}
    \vspace*{-1.5em}
\end{figure}
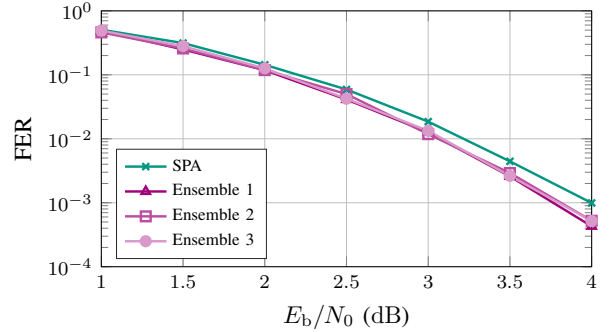
We compare different ensemble constructions at $\Tilde{K}=3$, i.e., the smallest size allowing the auxiliary paths to constitute an LC.
To this end, we consider the 5G LDPC code $\mathcal{C}_{5\mathrm{G}}(132,66)$ and sample ${N=\num{1000}}$ received words after an AWGN channel at an $E_{\mathrm{b}}/N_0$ of $4$ dB that stand-alone
 BP decoding using SPA with ${I_\mathrm{max}=32}$ could not decode, where $I_\mathrm{max}$ denotes the maximum number of iterations.
Note that for the 5G LDPC code, the PCM $\bm{H}$ has dimension $88\times154$ due to message bit puncturing \cite{5gstandard_2025_01}.
Using the greedy heuristic, we determine a coverage for $\Tilde{K}=3$ from ${c=3000}$ candidate rows which are sampled using a Bernoulli distribution with $p=\num{4.22}\%$ 
such that in average $\overline{\mathrm{w_H}(\bm{h})}=6.5$ (``Ensemble $1$''). 
Interestingly, this optimized set of subcodes---capable of decoding $591$ of the $\num{1000}$ frames---does not constitute an LC. 
Specifically, among other $\num{10000}$ randomly sampled codewords $11.88\%$ are not element of any of the $3$ subcodes.

Using $d_\mathrm{c}=6$, we run Algorithm~\ref{alg:append_row} to sample $\num{1000}$ tuples of rows $(\bm{h}_1,\bm{h}_2,\bm{h}_3)$ under two conditions: the rows constitute an LC (``Ensemble $2$'') and the rows do not constitute an LC (``Ensemble $3$'').
Note that in ``Ensemble $2$'', row $\bm{h}_3$ has row weight $\mathrm{w_H}(\bm{h}_3)=12$.
In contrast, ``Ensemble $3$'' does not yield an LC but instead runs Algorithm~\ref{alg:append_row} once per row of the tuple with $\mathcal{F}=[n]$, hence benefiting from $\mathrm{w_H}(\bm{h}_3)=6$.
  For both cases, we choose the tuple that decodes the most frames ($517$ and $526$ frames, respectively) from the gathered $N$ frames.

Fig.~\ref{fig:compare_construction} depicts the frame error rate (FER) over $E_{\mathrm{b}}/N_0$ of SCED given all three ensembles. All BP decodings use ${I_\mathrm{max}=32}$.
 Surprisingly, all ensembles yield a gain of about $0.25$ dB at an FER of $10^{-3}$ compared to stand-alone SPA decoding.  
 These results, which we also observed similarly for other codes, suggest that the slightly improved performance of ``Ensemble $1$'' (no LC, $\overline{\mathrm{w_H}(\bm{h}_3)}=\num{6.5}$) at an SNR of $4$ dB compared to ``Ensemble $2$'' (LC, $\mathrm{w_H}(\bm{h}_3)=\num{12}$) is not only because the third row has reduced weight. Hence, this indicates that, surprisingly, not all codewords must be element of at least one of the auxiliary subcodes.
Therefore, due to its simplicity, we will use the refined maximum-coverage heuristic used for ``Ensemble $1$'' with increased $c$ in Sec.~\ref{sec:results}.

\section{Comparison to Existing Schemes}

In \cite{mandelbaum2024endomorphisms}, we introduced EED which
involves $K$ parallel decoding paths. For each path, a distinct, not necessarily bijective endomorphism ${\tau_i:\mathcal{C}\rightarrow\mathcal{C}_i\subseteq\mathcal{C}}$ 
 is selected that maps codewords
onto a subcode $\mathcal{C}_i$, demonstrating some similarity to SCED. In decoding, EED performs three processing steps to mimic the effects of the endomorphisms in the LLR domain and to identify the most probable codeword within the set of possible preimages.
Similar to AED, EED samples different endomorphisms
to alter the noise representation.
In contrast, SCED generates diversity for ensemble decoding by sampling different subcodes and utilizing the different decoding behaviors of the respective subcode. 
Indeed, subcode decoding can also be applied to EED.
However, SCED avoids the processing of EED, which typically results in an information loss~\cite{mandelbaum2024endomorphisms}.

MBBP is another ensemble decoding scheme that improves the decoding performance of BP for algebraic codes in the short block length regime \cite{MBBP1}.
MBBP consists of $K$ parallel paths each incorporating a different PCM $\mathcal{C}$ of the code for BP. To this end, $K$ distinct and possibly overcomplete PCMs are generated and used to initialize the parallel decodings. However, unlike SCED, MBBP does not allow adding linearly independent rows but uses redundant representations of the kernel of the code. This complicates the search for suitable PCMs, as it relies on the NP-complete search for low-weight dual codewords\cite{kraft_setmatrix,intractabilityofminimumdistance}. Consequently, for SCED, the search for suitable PCMs is simplified because there exist many low-weight rows that are linearly independent of the rows of $\bm{H}$.

In \cite{stuttgart_ldpc_aed}, the authors show that in order to apply AED to quasi-cyclic (QC) LDPC codes the symmetry in the Tanner graph must be altered to use elements from the QC permutation automorphism group $\mathrm{Aut}_{\mathrm{QC}}$.
They propose three approaches for breaking the graph symmetry: adding rows, appending a linearly dependent row, or removing rows. The authors mainly consider the last method, due to its simplicity and because all approaches yield similar performance\cite{stuttgart_ldpc_aed}.
Interestingly, this approach constitutes decoding using an ambient code $\mathcal{C}_\mathrm{a}\supset \mathcal{C}$. Because the rows of a QC PCM are equivalent up to QC permutations, removing a row and applying a QC permutation is equivalent to removing one other row, i.e., $\forall i\in {[m]},\pi\in\mathrm{Aut}_{\mathrm{QC}}$, $\exists j\in{[m]}$ such that
\begin{equation}
    \pi(\bm{H}_{\sim i})=\bm{H}_{\sim j},\label{eq:qc_perm_eq}
\end{equation} where $\bm{H}_{\sim i}$ denotes $\bm{H}$ with the $i$th row removed.
Hence, we can gather $c=m$ candidates by removing every row of $\bm{H}\in\mathbb{F}_2^{m\times n}$ once, to generate an ensemble of PCMs of ambient codes that yield the same performance as AED when breaking the graph symmetry by removing a row. Thus, we refer to it as a row automorphism ensemble (R-AE). We use it as a comparison in Sec.~\ref{sec:results} to highlight the possible benefits of decoding on subcodes compared to decoding on ambient codes.
Note that (\ref{eq:inducing_subcode}) also typically breaks the graph symmetry.

\section{Results}\label{sec:results}

\subsection{Coverage Characteristics}\label{sec:coverage_characteristics}

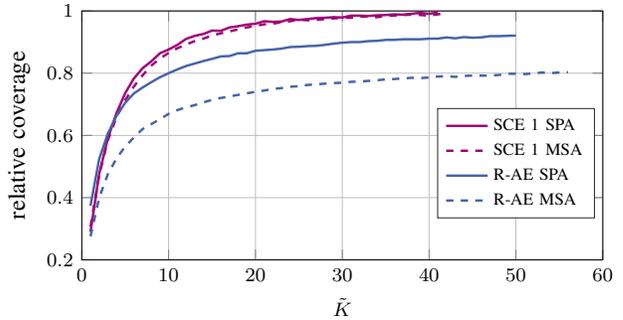
\begin{figure}
    \centering
    \begin{tikzpicture}[scale=0.92,spy using outlines={rectangle, magnification=2}]

\begin{axis}[%
width=.85\columnwidth,
height=3.6cm,
at={(0.758in,0.645in)},
scale only axis,
xmin=0,
xmax=60,
xlabel={$\Tilde{K}$},
ymin=0.2,
ymax=1,
yminorticks=true,
ylabel style={font=\color{white!15!black}},
ylabel={relative coverage},
axis background/.style={fill=white},
xmajorgrids,
ymajorgrids,
 legend style={at={(0.985,0.61)}, legend cell align=left, align=left, draw=white!15!black,font=\scriptsize}
]
    
\addplot[color=KITpurple,line width = 1pt, solid] 
table
{
1 0.29
2 0.476
3 0.591
4 0.672
5 0.736
6 0.783
7 0.816
8 0.838
9 0.864
10 0.876
11 0.891
12 0.91
13 0.92
14 0.926
15 0.937
16 0.936
17 0.947
18 0.951
19 0.956
20 0.959
21 0.967
22 0.965
23 0.966
24 0.974
25 0.971
26 0.974
27 0.977
28 0.978
29 0.98
30 0.979
31 0.985
32 0.985
33 0.984
34 0.986
35 0.987
36 0.989
37 0.988
38 0.991
39 0.994
40 0.991
41 0.995
};
\addlegendentry{SCE 1 SPA};

\addplot[color=KITpurple,line width = 1pt, dashed] 
table
{
1 0.3066933066933067
2 0.4725274725274725
3 0.5844155844155844
4 0.6593406593406593
5 0.7132867132867133
6 0.7572427572427572
7 0.7932067932067932
8 0.8241758241758241
9 0.8471528471528471
10 0.8651348651348651
11 0.8801198801198801
12 0.8921078921078921
13 0.903096903096903
14 0.9110889110889111
15 0.9210789210789211
16 0.9280719280719281
17 0.9340659340659341
18 0.9410589410589411
19 0.9470529470529471
20 0.952047952047952
21 0.957042957042957
22 0.9590409590409591
23 0.964035964035964
24 0.967032967032967
25 0.972027972027972
26 0.97002997002997
27 0.971028971028971
28 0.974025974025974
29 0.9760239760239761
30 0.978021978021978
31 0.978021978021978
32 0.981018981018981
33 0.983016983016983
34 0.98001998001998
35 0.984015984015984
36 0.985014985014985
37 0.987012987012987
38 0.988011988011988
39 0.987012987012987
40 0.985014985014985
41 0.989010989010989
42 0.993006993006993
};
\addlegendentry{SCE 1 MSA};

\addplot[color=KITblue,line width = 1pt] 
table
{
1 0.374
2 0.524
3 0.606
4 0.663
5 0.705
6 0.735
7 0.754
8 0.771
9 0.787
10 0.8
11 0.812
12 0.823
13 0.831
14 0.839
15 0.846
16 0.854
17 0.855
18 0.864
19 0.864
20 0.872
21 0.874
22 0.876
23 0.879
24 0.884
25 0.885
26 0.887
27 0.888
28 0.891
29 0.895
30 0.898
31 0.9
32 0.9
33 0.903
34 0.904
35 0.907
36 0.907
37 0.909
38 0.909
39 0.909
40 0.911
41 0.911
42 0.915
43 0.912
44 0.916
45 0.914
46 0.915
47 0.919
48 0.919
49 0.92
50 0.921
};
\addlegendentry{R-AE SPA};

\addplot[color=KITblue,dashed,line width = 1pt] 
table
{
1 0.2757242757242757
2 0.3996003996003996
3 0.4755244755244755
4 0.5264735264735265
5 0.5644355644355644
6 0.5934065934065934
7 0.6203796203796204
8 0.6363636363636364
9 0.6523476523476524
10 0.6693306693306693
11 0.6813186813186813
12 0.6903096903096904
13 0.6973026973026973
14 0.7062937062937062
15 0.7142857142857143
16 0.7202797202797203
17 0.7252747252747253
18 0.7302697302697303
19 0.7352647352647352
20 0.7402597402597403
21 0.7442557442557443
22 0.7492507492507493
23 0.7522477522477522
24 0.7562437562437563
25 0.7582417582417582
26 0.7622377622377622
27 0.7642357642357642
28 0.7662337662337663
29 0.7682317682317682
30 0.7692307692307693
31 0.7722277722277723
32 0.7732267732267732
33 0.7752247752247752
34 0.7772227772227772
35 0.7802197802197802
36 0.7812187812187812
37 0.7832167832167832
38 0.7852147852147852
39 0.7842157842157842
40 0.7862137862137862
41 0.7892107892107892
42 0.7902097902097902
43 0.7882117882117882
44 0.7912087912087912
45 0.7922077922077922
46 0.7922077922077922
47 0.7922077922077922
48 0.7942057942057942
49 0.7982017982017982
50 0.7982017982017982
51 0.7972027972027972
52 0.7992007992007992
53 0.8021978021978022
54 0.8001998001998002
55 0.8021978021978022
56 0.8031968031968032
};
\addlegendentry{R-AE MSA};

\end{axis}
\end{tikzpicture}
    \vspace*{-0.5em}
\caption{
Relative coverage as a function of the number of additional paths $\Tilde{K}$}
\label{fig:cov_plots}
\vspace*{-1.5em}
\end{figure}

For code $\mathcal{C}_{5\mathrm{G}}(132,66)$, we accumulate $N=\num{1000}$ frame errors from stand-alone decoding using SPA and MSA, respectively. For MSA, we employ a normalization factor of $\frac{3}{4}$.
Next, we construct $c=\num{35000}$ candidate rows, using a Bernoulli distribution
with $p=\num{4.22}\%$
as in Sec.~\ref{sec:comparison_ensemble}. We denote the collected set subcode ensemble (SCE).

We define the \emph{relative coverage} of a set of auxiliary paths as the ratio of the $N$ frames they can decode.
Note that, after collecting $N$ received words after the AWGN channel that BP with the original $\bm{H}$ can not decode, the relative coverage enables us to compare the performance of ensembles based on their performance on those $N$ frames without requiring exhaustive simulations of error rates for each $\Tilde{K}$.
Assuming that the first path employs $\mathcal{C}_1=\mathcal{C}$, Fig.~\ref{fig:cov_plots} depicts the relative coverage as a function of the number of additional paths $\Tilde{K}=K-1$.
We iteratively increase $\Tilde{K}$ and use the greedy heuristic to obtain optimized ensembles until reaching $\Tilde{K}_\mathrm{max}$,
the smallest $\Tilde{K}$ such that the selected candidates collectively cover all patterns that the union of all candidates can decode.
Further increasing $\Tilde{K}$ does not yield a larger relative coverage.

For both SPA and MSA, SCE achieves a very high relative coverage exceeding $98\%$ for sufficiently large $\Tilde{K}$. In contrast, R-AE only achieves a maximum relative coverage of $\num{92.1}\%$ and $\num{80.3}\%$ for SPA and MSA, respectively.
Notably, for the practically relevant MSA, R-AE shows significantly lower relative coverage compared to SCE.

\vspace*{-0.25em}
\subsection{Frame Error Rate Results}
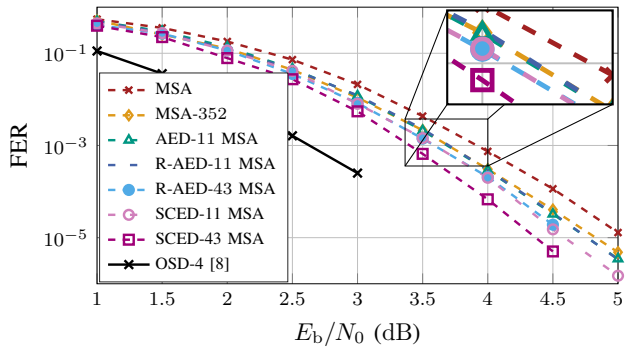
\begin{figure}
    \centering
    \begin{tikzpicture}[scale=0.92,spy using outlines={rectangle, magnification=2}]

\begin{axis}[%
width=.85\columnwidth,
height=4cm,
at={(0.758in,0.645in)},
scale only axis,
xmin=1,
xmax=5,
xlabel style={font=\color{white!15!black}},
xlabel={$E_{\mathrm{b}}/N_0$ ($\si{dB}$)},
ymode=log,
ymin=1e-06,
ymax=1,
yminorticks=true,
ylabel style={font=\color{white!15!black}},
ylabel={FER},
axis background/.style={fill=white},
xmajorgrids,
ymajorgrids,
legend style={at={(0.004,0.005)}, anchor=south west, legend cell align=left, align=left, draw=white!15!black,font=\scriptsize}
]

\addplot[color=KITred,dashed,line width = 1pt,mark=x, mark options={solid}]
table[row sep=crcr]{
 1.0  5.467980e-01\\
 1.5  3.553223e-01\\
 2.0  1.798365e-01\\
 2.5  7.278189e-02\\
 3.0  2.111038e-02\\
 3.5  4.332920e-03\\
 4.0  7.442976e-04\\
 4.5  1.148381e-04\\
 5.0  1.290430e-05\\
};
\addlegendentry{MSA};

\addplot[color=KITorange,dashed,line width = 1pt,mark=diamond, mark options={solid}]
table[row sep=crcr]{
 1.00  5.141388e-01\\
 1.50  2.551020e-01\\
 2.00  1.264223e-01\\
 2.50  4.380201e-02\\
3.00  1.090275e-02\\
 3.50  2.144956e-03\\
 4.00  3.055329e-04\\
 4.50  4.053300e-05\\
 5.00  4.830237e-06\\
};
\addlegendentry{MSA-$352$};

\addplot[color=KITgreen,loosely dashed,line width = 1pt,mark=triangle, mark options={solid}]
table[row sep=crcr]{
 1.00  4.926108e-01\\
 1.50  2.998501e-01\\
 2.00  1.362398e-01\\
 2.50  4.621072e-02\\
 3.00  1.192677e-02\\
 3.50  2.103359e-03\\
 4.00  2.896100e-04\\
4.50  3.245119e-05\\
 5.00  3.464770e-06\\
};
\addlegendentry{AED-$11$ MSA};

\addplot[color=KITblue,line width = 1pt,loosely dashed]
table[row sep=crcr]{
    1.00  4.878049e-01\\
    1.50  2.801120e-01\\
    2.00  1.250000e-01\\
    2.50  4.495392e-02\\
    3.00  1.221896e-02\\
    3.50  2.293710e-03\\
    4.00  2.946116e-04\\
    4.50  3.401602e-05\\
     5.00  3.355666e-06\\
};
\addlegendentry{R-AED-$11$ MSA};

\addplot[color=KITcyanblue,line width = 1pt,mark=*, dashed,mark options={solid}]
table[row sep=crcr]{
 1.00  4.106776e-01\\
 1.50  2.663116e-01\\
 2.00  1.152074e-01\\
 2.50  3.555556e-02\\
 3.00  8.322237e-03\\
3.50  1.407559e-03\\
4.00  2.073170e-04\\
4.50  1.911417e-05\\
};
\addlegendentry{R-AED-$43$ MSA};

\addplot[color=KITpurple!50,line width = 1pt,mark=o, mark options={solid},loosely dashed]
table[row sep=crcr]{
 1.00  4.464286e-01\\
 1.50  2.743484e-01\\
 2.00  1.147447e-01\\
 2.50  4.145078e-02\\
 3.00  8.008329e-03\\
 3.50  1.445713e-03\\
 4.00  1.989583e-04\\
4.50  1.514168e-05\\
 5.00  1.494075e-06\\
};
\addlegendentry{SCED-11 MSA};

\addplot[color=KITpurple,dashed,line width = 1pt,mark=square, mark options={solid}]
table[row sep=crcr]{
 1.00  3.960396e-01\\
 1.50  2.242152e-01\\
 2.00  7.855460e-02\\
 2.50  2.734108e-02\\
 3.00  5.499945e-03\\
 3.50  6.574903e-04\\
 4.00  6.791694e-05\\
4.50  5.026574e-06\\
};
\addlegendentry{SCED-43 MSA};

\addplot[color=black,line width = 1pt, solid,mark=x, mark size=2.5pt, mark options={solid}]
table[col sep=comma]{
1.00, 1.120e-01
1.50, 3.609e-02
2.00, 9.891e-03
2.50, 1.623e-03
3.00, 2.514e-04
};
\label{plot:5g132_osd}
\addlegendentry{OSD-4 \cite{stuttgart_ldpc_aed}};

\coordinate (spypoint) at (axis cs:3.41,0.0005);
			\coordinate (spyviewer) at (axis cs:4,0.025);	
			\spy[width=2.2cm,height=1.25cm, thin, spy connection path={\draw(tikzspyonnode.south west) -- (tikzspyinnode.south west);\draw (tikzspyonnode.south east) -- (tikzspyinnode.south east);
			\draw (tikzspyonnode.north west) -- (tikzspyinnode.north west);\draw (tikzspyonnode.north east) -- (intersection of  tikzspyinnode.north east--tikzspyonnode.north east and tikzspyinnode.south east--tikzspyinnode.south west);
			;}] on (spypoint) in node at (spyviewer);
		\coordinate (a) at ($(axis cs:-10.8/1.4,-0.12)+(spyviewer)$);
		\coordinate[label={[font=\small,text=black]right:$10^{-2}$}] (b) at ($(axis cs:+10.8/1.4,-0.12)+(spyviewer)$);
\end{axis}

\end{tikzpicture}%
    \vspace*{-0.5em}
    \caption{ Decoder performances for the 5G LDPC code $\mathcal{C}_{5\mathrm{G}}(132,66)$. }
    \label{fig:res_5g}
    \vspace*{-1em}
\end{figure}
\begin{figure}
    \centering
    \begin{tikzpicture}[scale=0.92,spy using outlines={rectangle, magnification=2}]

\begin{axis}[%
width=.85\columnwidth,
height=4cm,
at={(0.758in,0.645in)},
scale only axis,
xmin=1,
xmax=3,
xlabel style={font=\color{white!15!black}},
xlabel={$E_{\mathrm{b}}/N_0$ ($\si{dB}$)},
ymode=log,
ymin=1e-5,
ymax=1,
yminorticks=true,
ylabel style={font=\color{white!15!black}},
ylabel={FER},
axis background/.style={fill=white},
xmajorgrids,
ymajorgrids,
legend style={at={(0.03,0.03)}, anchor=south west, legend cell align=left, align=left, draw=white!15!black,font=\scriptsize}
]

\addplot[color=KITcyanblue,line width = 1pt,mark=x, mark options={solid}]
table[row sep=crcr]{
 1.00  4.184100e-01\\
 1.50  1.296176e-01\\
 2.00  1.907942e-02\\
 2.50  1.464976e-03\\
 3.00  7.492838e-05\\
 3.50  5.033615e-06\\
 4.00  4.250000e-07\\
};
\addlegendentry{SPA};

\addplot[color=KITpurple,line width = 1pt,mark=o, mark options={solid}]
table[row sep=crcr]{
 1.00  3.929273e-01\\
 1.50  9.832842e-02\\
 2.00  1.204602e-02\\
 2.50  5.844501e-04\\
 3.00  1.785529e-05\\
};
\addlegendentry{SCED-$11$ SPA};

\addplot[color=KITred,dashed,line width = 1pt,mark=triangle, mark options={solid}]
table[row sep=crcr]{
 1.00  5.856515e-01\\
 1.50  2.240896e-01\\
 2.00  3.575579e-02\\
 2.50  2.561262e-03\\
 3.00  1.663773e-04\\
 3.50  1.981652e-05\\
 4.00  2.450000e-06\\
};
\addlegendentry{MSA};

\addplot[color=KITorange,line width = 1pt,dashed,mark=square, mark options={solid}]
table[row sep=crcr]{
 1.00  4.914005e-01\\
 1.50  1.669449e-01\\
 2.00  2.333178e-02\\
 2.50  1.127605e-03\\
 3.00  3.520119e-05\\
};
\addlegendentry{SCED-$11$ MSA};

\end{axis}

\end{tikzpicture}%
    \vspace*{-0.5em}
         \caption{ Decoder performances for the code $\mathcal{C}_{\mathrm{irPEG}}(504,252)$ from \cite{mackay_codes_files}.}
         \label{fig:res_peg}
         \vspace*{-1.5em}
\end{figure}
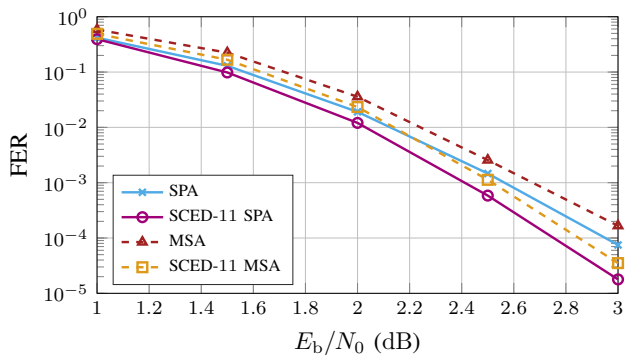
For BP decoding, we use an early stopping criterion if the current hard decision of the VNs fulfills $\hat{\bm{x}}\in\mathcal{C}$.
Let ${\lambda_i\leq I_\mathrm{max}}$ denote the actual number of iterations of BP decoding of the $i$th path, $i\in[K]$.
Assuming that all decodings of an ensemble decoding scheme are executed in parallel, we define the \emph{latency} as $\max_{i\in[K]} \lambda_i $  and the \emph{complexity} as $\sum_{i\in[K]} \lambda_i$. This notion of complexity is reasonable because the number of rows of the different PCMs is comparable.
Next, considering a target FER of $10^{-3}$, we analyze the performance of SCED for the QC LDPC code 
$\mathcal{C}_{5\mathrm{G}}(132,66)$ and the irregular LDPC code constructed using progressive edge growth
$\mathcal{C}_{\mathrm{irPEG}}(504,252)$ from \cite{mackay_codes_files} with unknown automorphism group. Hence, AED is not directly applicable for the code $\mathcal{C}_{\mathrm{irPEG}}(504,252)$.

To evaluate the performance of SCED, we perform Monte-Carlo simulations using a binary input AWGN channel collecting at least $\num{200}$ frame errors per data point. For consistency with \cite{stuttgart_ldpc_aed}, all BP decodings use ${I_\mathrm{max}=32}$ unless stated otherwise. 
For the codes $\mathcal{C}_{5\mathrm{G}}$ and $\mathcal{C}_{\mathrm{irPEG}}$, we generate ${c=\num{35000}}$ candidates paths with the row entries sampled using a Bernoulli distribution with $p=\num{4.22}\%$ and $p={1.29}\%$, respectively.
The notation SCED-$K$ refers to SCED using 
a total of $K$ paths, i.e., ${\Tilde{K}=K-1}$ auxiliary subcodes chosen from the $c$ candidates using the heuristic maximum-coverage combined with the decoding on $\bm{H}$.
Fig.~\ref{fig:res_5g}-\ref{fig:res_peg} show the FER over $E_{\mathrm{b}}/N_0$ for SCED of both codes with varying ensemble sizes, compared to the respective stand-alone BP decodings. 

In Fig.~\ref{fig:res_5g}, we also include the performance of AED-$11$, implemented according to \cite{stuttgart_ldpc_aed}.
Furthermore, we depict the performance of ensemble decoding using $10$ and $43$ paths gathered from R-AE using the greedy heuristic combined with decoding on $\bm{H}$, denoted as R-AED-$11$ and R\nobreakdash-AED\nobreakdash-$43$, respectively. We also provide the performance of ordered statistics decoding with order $4$ (OSD-$4$) from \cite{stuttgart_ldpc_aed} as an estimate of the ML performance.
Note that, as expected due to (\ref{eq:qc_perm_eq}), R-AED-$11$ and AED-$11$ yield identical performance.

Reflecting the higher relative coverage, SCED consistently yields gains compared to stand-alone BP decoding and AED with equal worst-case latency.
In Fig.~\ref{fig:res_5g}, SCED-$11$ yield gains of $0.3$ dB and $0.1$ dB compared to MSA and AED-$11$, respectively.
Notably, in Fig.~\ref{fig:res_5g}, while offering significantly reduced latency, SCED-$11$ achieves a gain of $0.1$~dB compared to equal-complexity stand-alone MSA using ${I_\mathrm{max}=352}$, denoted as MSA-$352$.
To evaluate the full potential of SCED, we also consider configurations using $K_\mathrm{max}=43$ auxiliary paths. For the 5G LDPC code, SCED-$43$ significantly outperforms R-AED-$43$ and reduces the gap to OSD-$4$ to $0.8$ dB.

Finally, for the code $\mathcal{C}_{\mathrm{irPEG}}(504,252)$ the automorphism group is unknown. Nevertheless, %
SCED ensembles can be designed straightforwardly based solely on the knowledge of the PCM $\bm{H}$.
Fig.~\ref{fig:res_peg} demonstrates that SCED-$11$ yield gains of approximately $0.2$~dB compared to both MSA and SPA.
\vspace*{-0.25em}
\subsection{Average Latency \& Qualitative Results}

As demonstrated, SCED yields a lower worst-case latency compared to AED and stand-alone decoding.
Yet, our simulations reveal that the average number of iterations of the BP decoding on the proper subcodes is increased compared to the first path, i.e., when decoding on $\bm{H}$. 
This is expected, as in an average of $50\%$ of the cases, the decoding of a proper subcode ${\mathcal{C}_i\subset \mathcal{C}}$ attempts to decode a codeword  ${\bm{x}\in \mathcal{C}\setminus \mathcal{C}_i}$.

Interestingly, when considering the average number of iterations for codewords that belong to the respective subcode, we even observe a decrease in the average number of iterations relative to the first path. Hence, we expect that introducing stopping mechanisms as in \cite{7360541} can maintain the error correction capabilities and reduce the average number of iterations.

In \cite[Table II]{krieg2024comparativestudyensembledecoding_arxiv}, the authors qualitatively compare the decoding gains of various ensemble decoding schemes for BP decoding with their requirements on the code and decoding structure. Table~\ref{table:qualitatively_comparison} extends this comparison to SCED, demonstrating that SCED achieves significant coding gains without imposing any requirements on the code and decoding structure.
\begin{table}
        \centering
        \caption{Qualitative comparison of ensemble decoding methods for BP decoding taken from 
        \vspace*{-0.5em}
        \cite{krieg2024comparativestudyensembledecoding_arxiv} in comparison to SCED.}\label{table:qualitatively_comparison}
        \begin{tabular}{cccc}
                \toprule
                Decoder &   Code Requirements & Decoder Requirements & Gain \\
                \midrule                
                MBBP& $\gg m$ Min. Weight Checks & -- & $++$ \\
                AED& Code Automorphisms & Non-Equivariance& $++$ \\
                SED& -- & Layered Decoder & $++$\\
                NED, SBP& -- & -- & $+$ \\
                \textcolor{KITgreen}{SCED}& \textcolor{KITgreen}{--} & \textcolor{KITgreen}{--} & $\color{KITgreen}{++}$ \\
                \bottomrule
        \end{tabular}
        \vspace*{-1.5em}
\end{table}
\vspace*{-1.5em}
\section{Conclusion}
In this work, we introduced SCED, an ensemble decoding scheme that leverages multiple decodings on subcodes of the original code. We discussed the concept of LCs for SCED and observed that, for BP decoding, effective ensembles can be sampled without ensuring that the auxiliary paths constitute an LC.
Our results demonstrate that for LDPC codes and BP decoding, SCED achieves improved decoding performance compared to both stand-alone decoding and AED. Notably, SCED does not rely on specific knowledge of the structure of the code and decoding, enabling the straightforward construction of good ensembles and making it easily adaptable to various codes.
Extending SCED to other code families, e.g., polar codes, is part of our ongoing research.

\IEEEtriggeratref{11}

\newpage

\appendix

We first provide Lemma~\ref{lemma:proof_of_existence} showing the existence of LCs consisting of $3$ proper subcodes, whose proof contains an important construction used in the upcoming proof of Theorem~\ref{theorem:lc_k-1_dim}:

\begin{lemma}\label{lemma:proof_of_existence}
    Let $\mathcal{C}\subseteq\mathbb{F}_2^{n}$ be a code constituting a $k$-dimensional vector space over $\mathbb{F}_2$. Then there exist $3$ proper $(k-1)$-dimensional subcodes $\mathcal{C}_1, \mathcal{C}_2, \mathcal{C}_3$, such that
    \begin{equation}\label{lemma_union}
            \mathcal{C}_1\cup \mathcal{C}_2 \cup \mathcal{C}_3=\mathcal{C}.
    \end{equation}
\end{lemma}

\begin{proof}[Proof of Lemma \ref{lemma:proof_of_existence}] Consider the basis $\{\bm{b}_1,\bm{b}_2\ldots,\bm{b}_k\}$ of $\mathcal{C}$, e.g., given by the rows of the generator matrix. Then, the $(k-1)$-dimensional subcodes $\mathcal{C}_1, \mathcal{C}_2, \mathcal{C}_3$ with bases
\begin{align*}
    \mathcal{B}_1&=\{\bm{b}_1,\bm{b}_3\ldots,\bm{b}_k\},\\
    \mathcal{B}_2&=\{\bm{b}_2,\bm{b}_3,\bm{b}_4\ldots,\bm{b}_k\},\\
    \mathcal{B}_3&=\{\bm{b}_1+\bm{b}_2,\bm{b}_3,\bm{b}_4\ldots,\bm{b}_k\},\\
\end{align*}
respectively, fulfill (\ref{lemma_union}).
\end{proof}

With a method for constructing LCs at hand, we can now prove Theorem \ref{theorem:lc_k-1_dim}:

 \begin{proof}[Proof of Theorem \ref{theorem:lc_k-1_dim}]
Let $\mathcal{C}$ be a binary linear code with PCM $\bm{H}\in \mathbb{F}_2^{m\times n}$ and let $\bm{h}_1,\bm{h}_2\in \mathbb{F}^{1\times n}$ be two row vectors that are linearly independent of the rows of $\bm{H}$, but not necessarily mutually independent. Since $\bm{h}_1, \bm{h}_2$ are linearly independent of the rows of $\bm{H}$, appending each of them to $\bm{H}$ according to (\ref{eq:inducing_subcode}) yields two PCMs denoted as $\bm{H}_1,\bm{H}_2$, respectively, which induce two proper subcodes $\mathcal{C}_1 \subset\mathcal{C},\mathcal{C}_2 \subset\mathcal{C}$, respectively. Note that the $\mathcal{C}_1$ and $\mathcal{C}_2$ are not necessarily distinct. We distinguish two cases. 

\textbf{Case 1:} Let $\mathcal{C}_1=\mathcal{C}_2$. In this case, appending $\bm{h}_2$ to $\bm{H}_1$ does not change the null space of $\bm{H}_1$, i.e., $\bm{h}_2$ is linearly dependent of the rows of
$\bm{H}_1$. Due to the assumption that both $\bm{h}_1$ and $\bm{h}_2$ are linearly independent of the rows of $\bm{H}$ and because appending both of them to $\bm{H}$ only increases the rank of the matrix by $1$ compared to $\bm{H}$, it follows that $\bm{h}_3=\bm{h}_1+\bm{h}_2$ is linearly dependent of the rows of $\bm{H}$. %
Hence, appending $\bm{h}_3$ to $\bm{H}$ according to (\ref{eq:inducing_subcode}) induces the subcode $\mathcal{C}_3=\mathcal{C}$ and, trivially, $\{\mathcal{C}_1,\mathcal{C}_2,\mathcal{C}_3\}$ constitute an LC.

\textbf{Case 2:} Let $\mathcal{C}_1\neq\mathcal{C}_2$. Thus, ${\bm{h}_1\neq\bm{h}_2}$, since ${\bm{h}_1=\bm{h}_2}$ would result in Case 1.
Note that the assumption $\mathcal{C}_1\neq\mathcal{C}_2$ is equivalent to $\bm{h}_3=\bm{h}_1+\bm{h}_2$ being linearly independent of the rows of $\bm{H}$ since otherwise appending $\bm{h}_2$ to $\bm{H}_1$ would not increase the rank of $\bm{H}_1$, i.e., $\mathcal{C}_1=\mathcal{C}_2$.
Furthermore, since both subcodes are proper subcodes of dimension $k-1$, i.e., $|\mathcal{C}_1|=|\mathcal{C}_2|$, that are distinct, there exist codewords $\bm{b}_1\in\mathcal{C}_1\setminus \mathcal{C}_2\subset \mathcal{C}$ and  $\bm{b}_2\in\mathcal{C}_2\setminus \mathcal{C}_1\subset \mathcal{C}$. By construction, $\bm{b}_1$, $\bm{b}_2$ 
have the properties that\footnote{Note that $\bm{h}_i$ are row vectors, whereas $\bm{b}_i$ are column vectors.} 
\begin{align}
&\bm{H}\bm{b}_1=\bm{0},    &\bm{H}\bm{b}_2=\bm{0},\nonumber\\
&\bm{h}_1 \bm{b}_1=0,  &\bm{h}_2\bm{b}_2=0,\label{eq:properties_of_bases}\\
&\bm{h}_2 \bm{b}_1=1,  &\bm{h}_1\bm{b}_2=1\nonumber,
\end{align}
 
 since, otherwise,
$$
\begin{pmatrix}
    \bm{H} \\ \bm{h}_2
\end{pmatrix}
\bm{b}_1
=
\bm{0}
$$
contradicting the choice of $\bm{b}_1$ and, similarly, for $\bm{b}_2$.

Since $\bm{b}_1, \bm{b}_2\in\mathcal{C}$ and since $\bm{b}_1\neq\bm{b}_2$ which implies that they are linearly independent, there exists a basis 
\begin{align*}
    \mathcal{B} 
    &=
    \{\bm{b}_1,\bm{b}_2, \bm{b}_3\ldots,\bm{b}_k\}
\end{align*}
for $\mathcal{C}$ such that
\begin{align*}
\mathcal{B}_1&=\{\bm{b}_1,\bm{b}_3\ldots,\bm{b}_k\}\\    
\mathcal{B}_2&=\{\bm{b}_2,\bm{b}_3,\bm{b}_4\ldots,\bm{b}_k\}
\end{align*}
form a basis of $\mathcal{C}_1$ and $\mathcal{C}_2$, respectively. Now, choosing ${\bm{h}_3=\bm{h}_1+\bm{h}_2\neq\bm{0}}$ and appending $\bm{h}_3$ to $\bm{H}$ induces a proper subcode $\mathcal{C}_3\subset\mathcal{C}$.

It remains to show that $\mathcal{C}_1\cup\mathcal{C}_2\cup\mathcal{C}_3=\mathcal{C}$. First, note that all codewords composed of linear combinations within $\mathcal{B}_1$ or $\mathcal{B}_2$ alone are already covered by $\mathcal{C}_1$ or $\mathcal{C}_2$. Hence, the remaining codewords $\bm{x}\in \mathcal{C}\setminus (\mathcal{C}_1\cup\mathcal{C}_2)$
are of the form
\begin{equation}\label{eq:not_in_c1_c2}
    \bm{x}=\bm{b}_1+\bm{b}_2+\sum\limits_{i=3}^k \alpha_i \bm{b}_i
\end{equation}
 and are, thus, covered by $\mathcal{C}_3$ since (\ref{eq:not_in_c1_c2}) yields:
$$
\begin{pmatrix}
    \bm{H} \\ \bm{h}_3
\end{pmatrix}
\bm{x}
=
\begin{pmatrix}
    \bm{H} \\ \bm{h}_1+\bm{h}_2
\end{pmatrix}
\cdot \left( \bm{b}_1+\bm{b}_2+\sum\limits_{i=3}^k \alpha_i \bm{b}_i\right)
$$
The first $m$ rows are equal to zero since $\bm{x}\in\mathcal{C}$, and the last row becomes:
\begin{align*}
&\quad 
(\bm{h}_1+\bm{h}_2)\cdot \left( \bm{b}_1+\bm{b}_2+\sum\limits_{i=3}^k \alpha_i \bm{b}_i\right)
\\
&\stackrel{(a)}{=}
0 + 1 + 0+ 1 + 0 + 0 
= 
0
\end{align*}
where $(a)$ is due to expanding the terms and using (\ref{eq:properties_of_bases}).
\end{proof}

\emph{Remark:} Note that, for practical codes, random sampling of rows $\bm{h}_1$ and $\bm{h}_2$ typically results in the second case.


\begin{thebibliography}{10}
\providecommand{\url}[1]{#1}
\csname url@samestyle\endcsname
\providecommand{\newblock}{\relax}
\providecommand{\bibinfo}[2]{#2}
\providecommand{\BIBentrySTDinterwordspacing}{\spaceskip=0pt\relax}
\providecommand{\BIBentryALTinterwordstretchfactor}{4}
\providecommand{\BIBentryALTinterwordspacing}{\spaceskip=\fontdimen2\font plus
\BIBentryALTinterwordstretchfactor\fontdimen3\font minus
  \fontdimen4\font\relax}
\providecommand{\BIBforeignlanguage}[2]{{%
\expandafter\ifx\csname l@#1\endcsname\relax
\typeout{** WARNING: IEEEtran.bst: No hyphenation pattern has been}%
\typeout{** loaded for the language `#1'. Using the pattern for}%
\typeout{** the default language instead.}%
\else
\language=\csname l@#1\endcsname
\fi
#2}}
\providecommand{\BIBdecl}{\relax}
\BIBdecl

\bibitem{MBBP1}
T.~Hehn, J.~B. Huber, O.~Milenkovic, and S.~Laendner, ``Multiple-bases
  belief-propagation decoding of high-density cyclic codes,'' \emph{IEEE Trans.
  Commun.}, vol.~58, no.~1, pp. 1--8, January 2010.

\bibitem{MBBP2}
T.~Hehn, J.~B. Huber, S.~Laendner, and O.~Milenkovic, ``Multiple-bases
  belief-propagation for decoding of short block codes,'' in \emph{Proc. {IEEE}
  Int. Symp. Inf. Theory (ISIT)}, Nice, France, Jun. 2007.

\bibitem{MBBP3_withLeaking}
T.~Hehn, J.~B. Huber, P.~He, and S.~Laendner, ``Multiple-bases
  belief-propagation with leaking for decoding of moderate-length block
  codes,'' in \emph{Proc. Int. ITG Conf. on Source and Channel Coding (SCC)},
  Ulm, Germany, Jan. 2008.

\bibitem{AED_RMcodes}
M.~Geiselhart, A.~Elkelesh, M.~Ebada, S.~Cammerer, and S.~ten Brink,
  ``Automorphism ensemble decoding of {Reed-Muller} codes,'' \emph{IEEE Trans.
  Commun.}, vol.~69, no.~10, pp. 6424--6438, Oct. 2021.

\bibitem{stuttgart_ldpc_aed}
M.~Geiselhart, M.~Ebada, A.~Elkelesh, J.~Clausius, and S.~ten Brink,
  ``Automorphism ensemble decoding of quasi-cyclic {LDPC} codes by breaking
  graph symmetries,'' \emph{IEEE Commun. Lett.}, vol.~26, no.~8, pp.
  1705--1709, Aug. 2022.

\bibitem{geiselhart2023ratecompatible}
M.~Geiselhart, J.~Clausius, and S.~ten Brink, ``Rate-compatible polar codes for
  automorphism ensemble decoding,'' in \emph{Proc. Int. Symp. on Topics in
  Coding (ISTC)}, Brest, France, 2023.

\bibitem{mandelbaum2023generalized}
J.~Mandelbaum, H.~Jäkel, and L.~Schmalen, ``Generalized automorphisms of
  channel codes: Properties, code design, and a decoder,'' in \emph{Proc. Int.
  Symp. on Topics in Coding (ISTC)}, Brest, France, Sept. 2023.

\bibitem{MacWilliamsSloane}
F.~J. MacWilliams and N.~J.~A. Sloane, \emph{The Theory of Error Correcting
  Codes}, ser. North-Holland Mathematical Library.\hskip 1em plus 0.5em minus
  0.4em\relax Elsevier, 1977.

\bibitem{8464884}
F.~Gensheimer, T.~Dietz, S.~Ruzika, K.~Kraft, and N.~Wehn, ``Improved
  maximum-likelihood decoding using sparse parity-check matrices,'' in
  \emph{Proc. Int. Conf. on Telecommun. (ICT)}, Saint Malo, France, Jun. 2018.

\bibitem{mandelbaum2024endomorphisms}
J.~Mandelbaum, S.~Miao, H.~J{\"a}kel, and L.~Schmalen, ``Endomorphisms of
  linear block codes,'' in \emph{Proc. {IEEE} Int. Symp. Inf. Theory (ISIT)},
  Athens, Greece, Jul. 2024.

\bibitem{clark_covering_numbers}
P.~L. Clark, ``Covering numbers in linear algebra,'' \emph{The American
  Mathematical Monthly}, vol. 119, pp. 65--67, Jan. 2012.

\bibitem{kraft_setmatrix}
K.~Kraft, M.~Hermann, O.~Griebel, and N.~Wehn, ``Ensemble belief propagation
  decoding for short linear block codes,'' in \emph{Proc. Int ITG Workshop on
  Smart Antennas (WSA) and Conference on Systems, Communications, and Coding
  (SCC)}, Braunschweig, Germany, Feb. 2023.

\bibitem{10.1007/3-540-48777-8_2}
A.~A. Ageev and M.~I. Sviridenko, ``Approximation algorithms for maximum
  coverage and max cut with given sizes of parts,'' in \emph{In Proc. Integer
  Programming and Combinatorial Optimization (IPCO)}, Graz, Austria, Jun. 1999.

\bibitem{krieg2024comparativestudyensembledecoding_arxiv}
\BIBentryALTinterwordspacing
F.~Krieg, J.~Clausius, M.~Geiselhart, and S.~ten Brink, ``A comparative study
  of ensemble decoding methods for short length ldpc codes,'' 2024. [Online].
  Available: \url{https://arxiv.org/abs/2410.23980}
\BIBentrySTDinterwordspacing

\end{thebibliography}


\begin{thebibliography}{10}
\providecommand{\url}[1]{#1}
\csname url@samestyle\endcsname
\providecommand{\newblock}{\relax}
\providecommand{\bibinfo}[2]{#2}
\providecommand{\BIBentrySTDinterwordspacing}{\spaceskip=0pt\relax}
\providecommand{\BIBentryALTinterwordstretchfactor}{4}
\providecommand{\BIBentryALTinterwordspacing}{\spaceskip=\fontdimen2\font plus
\BIBentryALTinterwordstretchfactor\fontdimen3\font minus \fontdimen4\font\relax}
\providecommand{\BIBforeignlanguage}[2]{{%
\expandafter\ifx\csname l@#1\endcsname\relax
\typeout{** WARNING: IEEEtran.bst: No hyphenation pattern has been}%
\typeout{** loaded for the language `#1'. Using the pattern for}%
\typeout{** the default language instead.}%
\else
\language=\csname l@#1\endcsname
\fi
#2}}
\providecommand{\BIBdecl}{\relax}
\BIBdecl

\bibitem{MCT08}
T.~Richardson and R.~Urbanke, \emph{Modern Coding Theory}.\hskip 1em plus 0.5em minus 0.4em\relax Cambridge University Press, 2008.

\bibitem{8594709}
M.~Shirvanimoghaddam, M.~S. Mohammadi, R.~Abbas, A.~Minja, C.~Yue, B.~Matuz, G.~Han, Z.~Lin, W.~Liu, Y.~Li, S.~Johnson, and B.~Vucetic, ``Short block-length codes for ultra-reliable low latency communications,'' \emph{IEEE Commun. Mag.}, vol.~57, no.~2, pp. 130--137, Feb. 2019.

\bibitem{10571997}
S.~Miao, C.~Kestel, L.~Johannsen, M.~Geiselhart, L.~Schmalen, A.~Balatsoukas-Stimming, G.~Liva, N.~Wehn, and S.~T. Brink, ``Trends in channel coding for 6g,'' \emph{Proc. of the IEEE}, vol. 112, no.~7, pp. 653--675, Jul. 2024.

\bibitem{10274812}
M.~Geiselhart, F.~Krieg, J.~Clausius, D.~Tandler, and S.~ten Brink, ``{6G}: A welcome chance to unify channel coding?'' \emph{IEEE BITS the Information Theory Magazine}, pp. 1--12, Mar. 2023.

\bibitem{AED_RMcodes}
M.~Geiselhart, A.~Elkelesh, M.~Ebada, S.~Cammerer, and S.~ten Brink, ``Automorphism ensemble decoding of {Reed-Muller} codes,'' \emph{IEEE Trans. Commun.}, vol.~69, no.~10, pp. 6424--6438, Oct. 2021.

\bibitem{MBBP1}
T.~Hehn, J.~B. Huber, O.~Milenkovic, and S.~Laendner, ``Multiple-bases belief-propagation decoding of high-density cyclic codes,'' \emph{IEEE Trans. Commun.}, vol.~58, no.~1, pp. 1--8, Jan. 2010.

\bibitem{geiselhart2023ratecompatible}
M.~Geiselhart, J.~Clausius, and S.~ten Brink, ``Rate-compatible polar codes for automorphism ensemble decoding,'' in \emph{Proc. Int. Symp. on Topics in Coding (ISTC)}, Brest, France, Sept. 2023.

\bibitem{stuttgart_ldpc_aed}
M.~Geiselhart, M.~Ebada, A.~Elkelesh, J.~Clausius, and S.~ten Brink, ``Automorphism ensemble decoding of quasi-cyclic {LDPC} codes by breaking graph symmetries,'' \emph{IEEE Commun. Lett.}, vol.~26, no.~8, pp. 1705--1709, Aug. 2022.

\bibitem{8723089}
A.~Çağrı Arlı and O.~Gazi, ``Noise-aided belief propagation list decoding of polar codes,'' \emph{IEEE Commun. Lett.}, vol.~23, no.~8, pp. 1285--1288, Aug. 2019.

\bibitem{7360541}
P.~Schläfer, S.~Scholl, E.~Leonardi, and N.~Wehn, ``A new {LDPC} decoder hardware implementation with improved error rates,'' in \emph{In Proc. IEEE Jordan Conf. on Applied Electrical Engineering and Computing Technologies (AEECT)}, Amman, Jordan, Nov. 2015.

\bibitem{MBBP2}
T.~Hehn, J.~B. Huber, S.~Laendner, and O.~Milenkovic, ``Multiple-bases belief-propagation for decoding of short block codes,'' in \emph{Proc. {IEEE} Int. Symp. Inf. Theory (ISIT)}, Nice, France, Jun. 2007.

\bibitem{MBBP3_withLeaking}
T.~Hehn, J.~B. Huber, P.~He, and S.~Laendner, ``Multiple-bases belief-propagation with leaking for decoding of moderate-length block codes,'' in \emph{Proc. Int. ITG Conf. on Source and Channel Coding (SCC)}, Ulm, Germany, Jan. 2008.

\bibitem{mandelbaum2024endomorphisms}
J.~Mandelbaum, S.~Miao, H.~J{\"a}kel, and L.~Schmalen, ``Endomorphisms of linear block codes,'' in \emph{Proc. {IEEE} Int. Symp. Inf. Theory (ISIT)}, Athens, Greece, Jul. 2024.

\bibitem{mandelbaum2023generalized}
J.~Mandelbaum, H.~Jäkel, and L.~Schmalen, ``Generalized automorphisms of channel codes: Properties, code design, and a decoder,'' in \emph{Proc. Int. Symp. on Topics in Coding (ISTC)}, Brest, France, Sept. 2023.

\bibitem{krieg2024comparativestudyensembledecoding_arxiv}
\BIBentryALTinterwordspacing
F.~Krieg, J.~Clausius, M.~Geiselhart, and S.~ten Brink, ``A comparative study of ensemble decoding methods for short length {LDPC} codes,'' 2024. [Online]. Available: \url{https://arxiv.org/abs/2410.23980}
\BIBentrySTDinterwordspacing

\bibitem{intractabilityofminimumdistance}
A.~Vardy, ``The intractability of computing the minimum distance of a code,'' \emph{IEEE Trans. on Inf. Theory}, vol.~43, no.~6, pp. 1757--1766, Nov. 1997.

\bibitem{kraft_setmatrix}
K.~Kraft, M.~Hermann, O.~Griebel, and N.~Wehn, ``Ensemble belief propagation decoding for short linear block codes,'' in \emph{Proc. Int. ITG Workshop on Smart Antennas (WSA) and Conf. on Systems, Communications, and Coding (SCC)}, Braunschweig, Germany, Feb. 2023.

\bibitem{clark_covering_numbers}
P.~L. Clark, ``Covering numbers in linear algebra,'' \emph{The American Mathematical Monthly}, vol. 119, pp. 65--67, Jan. 2012.

\bibitem{10.1007/3-540-48777-8_2}
A.~A. Ageev and M.~I. Sviridenko, ``Approximation algorithms for maximum coverage and max cut with given sizes of parts,'' in \emph{In Proc. Integer Programming and Combinatorial Optimization (IPCO)}, Graz, Austria, Jun. 1999.

\bibitem{5gstandard_2025_01}
\BIBentryALTinterwordspacing
\emph{{5G; NR; Physical layer; Data (Release 17)}}, {3rd Generation Partnership Project (3GPP)} Std. TS 38.212, Mar. 2023, version 17.7.0. [Online]. Available: \url{https://portal.3gpp.org/desktopmodules/Specifications/SpecificationDetails.aspx?specificationId=3214}
\BIBentrySTDinterwordspacing

\bibitem{mackay_codes_files}
D.~J.~C. MacKay, ``Information theory, inference, and learning algorithms: Codes files,'' \url{https://www.inference.org.uk/mackay/CodesFiles.html}, accessed: 2024-12-22.

\end{thebibliography}
\end{document}